\documentclass[aps,pra,preprint,groupedaddress,nofootinbib]{revtex4-2}

\usepackage{amsmath}
\usepackage{amsfonts}
\usepackage{amsthm}
\usepackage{mathtools}
\usepackage{listings}
\usepackage{array}
\usepackage{graphicx}
\usepackage{tikz}
\usepackage[hidelinks]{hyperref}
\usepackage[capitalise]{cleveref}
\newcommand{\ket}[1]{\left| #1 \right\rangle}
\newcommand{\bra}[1]{\left\langle #1 \right|}
\newcommand{\braket}[2]{\left\langle #1 \middle| #2 \right\rangle}
\newcommand{\kl}[1]{\ket{\lambda_{#1}}}
\newcommand{\bl}[1]{\bra{\lambda_{#1}}}
\newcommand{\bk}[2]{\braket{\lambda_{#1}}{\lambda_{#2}}}
\newcommand{\svl}[4]{S^{#1_{#2}}VS^{#1_{#3}}V\!\dots S^{#1_{#4}}V}
\newcommand{\sv}[3]{S^{#1_{#2}}V\!\dots S^{#1_{#3}}V}
\newcommand{\vs}[3]{VS^{#1_{#2}}\!\dots VS^{#1_{#3}}}

\DeclarePairedDelimiter{\ceil}{\lceil}{\rceil}
\DeclarePairedDelimiter{\floor}{\lfloor}{\rfloor}

\makeatletter
\def\th@plain{%
  \thm@notefont{}
  \itshape 
}
\def\th@definition{%
  \thm@notefont{}
  \normalfont 
}
\makeatother

\newtheorem{theorem}{Theorem}
\newtheorem{corollary}{Corollary}[theorem]
\newtheorem{lemma}{Lemma}
\theoremstyle{definition}

\newtheorem{definition}{Definition}

\newcommand{\diagram}[1]{\foreach \x in #1 {--++(0,\x)--++(1,0)}--cycle}
\newcommand{\diagraml}[1]{node[anchor=north west]{$(#1)$}\foreach \x in #1 {--++(0,\x)--++(1,0)}--cycle}
\newcommand{\drawgrid}[2]{\begin{scope}\begin{pgfinterruptboundingbox} 
\path #1 ++(0,-0.2) -- ++(21,21) -- ++(-0.4,0) -- ++(-21,-21) -- cycle [invclip]; \end{pgfinterruptboundingbox}
\draw[white,very thick] #1 \foreach \x [evaluate=\x as \z using \x - 1] in #2 {\ifnum \x=0 +(0,-0.1)--+(0,0.1) \fi \ifnum \x>1 \foreach \y in {1,...,\z} {++(0,1) +(-0.1,0)--+(0.1,0)} \fi \ifnum \x>0 ++(0,1) \fi ++(1,0)};\end{scope}}
\newcommand{\diagce}[4]{
	\draw #1 node[anchor=north west]{\hspace{-6pt}\begingroup\addtolength{\jot}{-0.75em}$\begin{aligned}&(#2)\\ &c,e=\textstyle#3,\textstyle#4\end{aligned}$\endgroup}
	\foreach \x in #2 {--++(0,\x)--++(1,0)}--cycle;
	\drawgrid{#1}{#2}}
\newcommand{\eequiv}[3]{\draw (#1,#2)++(-0.4,-0.2)--++(0,1)--++(5*#3-4.2+#2,0)--++(0,-1);}
\newcommand{\cequiv}[3]{\draw (#1,#2)++(-0.2,-0.2)--++(0,0.6)--++(5*#3-4.6+#2,0)--++(0,-0.6);}

\Crefname{section}{Sec.}{Secs.}

\begin{document}

\tikzset{invclip/.style={clip,insert path={{[reset cm]
      (-16383.99999pt,-16383.99999pt) rectangle (16383.99999pt,16383.99999pt)
    }}}}

\title{Explicit diagrammatic solution of normalised, nondegenerate Rayleigh-Schr\"odinger perturbation theory}

\author{Joel C. Pommerening}
\author{David P. DiVincenzo}
\thanks{Joel C. Pommerening developed the proofs, performed the calculations, made the figures, and wrote the manuscript. Both authors discussed the results, contributed to the literature search, verified the calculations, and revised the manuscript. David P. DiVincenzo supervised the project.}
\affiliation{Institute for Quantum Information, RWTH Aachen University, D-52056 Aachen, Germany}
\affiliation{Peter Gr\"unberg Institute, Theoretical Nanoelectronics, Forschungszentrum J\"ulich, D-52425 J\"ulich, Germany}
\affiliation{J\"ulich-Aachen Research Alliance (JARA), Fundamentals of Future Information Technologies, D-52425 J\"ulich, Germany}

\date{\today}

\begin{abstract}
We solve the coupled recurrence relations for eigenenergies and -vectors in nondegenerate Rayleigh-Schr\"odinger perturbation theory under the constraint that the approximate eigenvector be normalised to $1$ in every order.
The series can be expressed in terms of diagrams that were first introduced by C. Bloch [Nucl. Phys. \textbf{6}, 329 (1958)] for the degenerate, unnormalised case.
Normalisation increases the number of terms and introduces a nontrivial dependence on the diagrams' topology to the coefficients.
\end{abstract}

\maketitle

\section{Introduction}
Rayleigh-Schr\"odinger perturbation theory (RSPT)~\cite{Schroedinger26} is a simple, yet powerful tool for approximating Hamiltonian spectra and eigenfunctions. Its application is so ubiquitous that anyone who has ever done any quantum mechanics calculations is likely to have used it at some point.
Corrections to the eigenvalues and eigenvectors of an unperturbed problem are given as as a power series in a small perturbation.
In the standard textbook approach (e.g. \cite[Ch.~11]{Cohen86}) corrections are determined recursively as a function of all lower order terms.
Explicit expressions have long been known as well~\cite{Kato49,Kato50,Bloch58,Huby61,Salzman68,Silverstone70}, but, notably, not for the \emph{normalised} eigenfunctions.

Kato~\cite{Kato49,Kato50} gave the first explicit solution of (generally degenerate) RSPT.
Instead of choosing an arbitrary eigenbasis, he stated the results in terms of projectors onto (possibly still degenerate) eigenspaces.
Bloch~\cite{Bloch58} modified these projectors, reducing the number of terms by a factor of $2$.
He assumed a perturbation that completely lifts the degeneracy and concerned himself with the construction of an appropriate basis of the degenerate unperturbed eigenspace (``\emph{les `bonnes fonctions' non perturb{\'e}es}'').
Bloch also introduced the diagrammatic representation described below, as well as an alternative choice of \emph{bonnes fonctions} that allowed for the restriction to a subset of diagrams, called convex, further reducing the number of terms in order $n$ by a factor of $(n+1)/2$.

Earlier, following Brueckner~\cite{Brueckner55}, Goldstone~\cite{Goldstone57} used Feynman diagrams to explicitly write down corrections to the nondegenerate ground state of an interacting fermionic system.
Huby~\cite{Huby61} restated Bloch's results, in a form suggested by Brueckner \cite{Brueckner55}, where the same terms are constructed in a different way.
He can give explicit formulas for eigenvectors and not merely projectors because he considered the nondegenerate case.
These expressions for the eigenvectors were not normalised. 
Salzman~\cite{Salzman68} similarly focused on the nondegenerate, unnormalised case, and developed a new diagrammatic formalism, set up to collect equivalent terms.
This in  principle allows for a further reduction in the number of terms.
The rules he gave for constructing diagrams do not provide this reduction automatically however.
Equivalent terms still had to be collected together manually.
From the more mathematical direction, the equivalence of Bloch's diagram counting with that of the leaves of ordered trees can be found in the work of Stanley \cite{10.5555/2124415}. More recent surveys have related term counting in Rayleigh-Schr{\"o}dinger perturbation theory to other combinatorial objects \cite{https://doi.org/10.1002/qua.23201}.

Silverstone and Holloway~\cite{Silverstone70} derived alternative formulas for the nondegenerate eigenvalues and their unnormalised eigenvectors that formally lead to the least number of terms, however, at the price of evaluating a large number of derivatives.
Quantifying the number of terms in the resulting Silverstone-Holloway expression, beyond `more than the number of partitions of $n$ into positive integers,' is a non-trivial task.
More recently, Magesan and Gambetta~\cite{Magesan20} developed a formalism that preserves norms exactly by perturbatively series expanding the generator of a unitary operator.
For a given order, the canonical transformation of the Hamiltonian by that unitary is then in turn series expanded.
This method does not directly give an explicit series for eigenvectors and -values.
Bloch's original work still finds application in the context of effective Hamiltonians in Jordan and Farhi's arbitrary order perturbative gadgets~\cite{Jordan08}.

In the present work, we consider anew the perturbation of a nondegenerate eigenvalue in standard RSPT.
The phase and normalisation freedom of the eigenvector significantly influences the expansion.
To the best of our knowledge, here we give the first explicit solution which preserves the norm of the eigenvector at $1$ in every order.

In \cref{RSPT}, we briefly review nondegenerate RSPT and the main results of Bloch \cite{Bloch58} that we build upon.
Our new results are derived in \cref{main}.
We comment on their efficiency and how they can be improved in \cref{uniqueness}.
Finally, in \cref{example}, we focus more on the diagrammatic aspect and show how our results work in practice, going up to fourth order, and conclude in \cref{conclusio}.

\section{Rayleigh-Schr\"odinger perturbation theory \label{RSPT}}
\subsection{Recursive definition}
Given a Hamiltonian $H=H_0+\epsilon V$ parametrised by $\epsilon\in[0,1]$, we assume we can expand any of its eigenvalues $\lambda$, and the corresponding eigenvector $\ket\lambda$, in a power series in $\epsilon$,
\begin{equation}
\lambda=\sum_{n=0}^\infty\epsilon^n\lambda_n,\qquad\ket\lambda=\sum_{n=0}^\infty\epsilon^n\ket{\lambda_n},
\end{equation}
i.e. they satisfy
\begin{equation}\label{RSPTdef}
\left(H_0+\epsilon V\right)\sum_{n=0}^\infty\epsilon^n\ket{\lambda_n}=\sum_{n,m=0}^\infty\epsilon^{n+m}\ket{\lambda_m}\lambda_n.
\end{equation}
Sorting \cref{RSPTdef} by powers of $\epsilon$ and equating the coefficients gives in zeroth order $H_0\ket{\lambda_0}=\lambda_0\ket{\lambda_0}$.
Usually $H_0$ is chosen to be analytically diagonalisable; we call $\lambda_0$ and $\ket{\lambda_0}$ the unperturbed eigenenergies and -vectors, respectively.
Here we further assume that they are discrete and nondegenerate, and we define the complementary projectors
\begin{equation}
P_{\lambda}=\ket{\lambda_0}\bra{\lambda_0},\qquad Q_{\lambda}=1-P_{\lambda}.
\end{equation}
For nonzero powers $n\in\mathbb{N}$ of $\epsilon$, \cref{RSPTdef} gives
\begin{equation}\label{baseRR}
H_0\kl n+V\kl{n-1}=\sum_{m=0}^n\kl{n-m}\lambda_m.
\end{equation}
We can then consider the $P_{\lambda}$- and $Q_{\lambda}$-components of \cref{baseRR} separately to derive equations for $\lambda_n$ and $\kl n$ respectively.
For the energies we get
\begin{equation}\label{P0comp}
\lambda_n=\bl 0 V\kl{n-1}-\sum_{m=1}^{n-1}\bk{0}{n-m}\lambda_m.
\end{equation}
Note, here and throughout the paper, the convention
\begin{equation}\label{nullsum}
\sum_{n=N}^M A_n=0\quad\text{for any sequence }A_n\text{ if }M<N
\end{equation}
applies.
The $Q_{\lambda}$-component gives
\begin{equation}\label{Q0comp}
Q_{\lambda}\kl n=\frac{Q_{\lambda}}{\lambda_0-H_0}\left(V\kl{n-1}-\sum_{m=1}^{n-1}\kl{n-m}\lambda_m\right),
\end{equation}
where on the right-hand side we see appearing the reduced resolvent
\begin{equation}
S=\frac{Q_{\lambda}}{\lambda_0-H_0}=Q_{\lambda}\frac{1}{\lambda_0-H_0}Q_{\lambda}=\sum_{\lambda^\prime\neq\lambda}\frac{\ket{\lambda^\prime_0}\bra{\lambda^\prime_0}}{\lambda_0-\lambda^\prime_0}
\end{equation}
which, since by assumption $\lambda_0$ is nondegenerate, is well-defined. For the sake of a compact notation, there is no index $\lambda$ on $S$, but it should be remembered as implicit. For later reference we also define powers of $S$, where it will be convenient to define $S^0$ separately \cite{Kato49}
\begin{equation}
S^0=-P_{\lambda},\qquad S^k=\frac{Q_{\lambda}}{\left(\lambda_0-H_0\right)^k}\quad\text{for }k\in\mathbb{N}.
\end{equation}

Clearly $\bk0n$ is not constrained by \cref{Q0comp} or \cref{P0comp}, and by extension \cref{RSPTdef}.
The simplest choice, and one employed by many authors~\cite{Huby61,Salzman68,Silverstone70}
, is $\bk0n=0$.
But it can be more convenient to use this degree of freedom to normalise the eigenvector to $1$ in every order, i.e.
\begin{equation}\label{norm}
\sum_{n,m=0}^N\epsilon^{n+m}\bk n m=1+\mathcal{O}\left(\epsilon^{N+1}\right)\qquad\forall\,N\in\mathbb{N}_0.
\end{equation}
This way the calculated eigenvectors immediately form an orthonormal basis (up to higher order terms) and can be used straightforwardly to calculate expectation values without having to manually renormalise.
\Cref{norm} requires the unperturbed eigenvector to be normalised, $\bk00=1$, and fixes the real part of $\bk0n$.
Choosing to set the imaginary part to $0$, we arrive at
\begin{equation}\label{zeron}
\bk0n=-\frac12\sum_{m=1}^{n-1}\bk m{n-m}.
\end{equation}
This is not a unique phase choice (see \cref{uniqueness}), though it is the conventional one~\cite[Ch.~11]{Cohen86}. 
In \cref{RSPTRR}, we collect \cref{P0comp,Q0comp,zeron}.
It is not a new result but rarely stated explicitly for arbitrary orders.
\begin{theorem}[Cohen-Tannoudji et al.]\label{RSPTRR}
The sequences $\lambda_n$ and $\kl n$, $n\in\mathbb{N}$ that satisfy the coupled recurrence relations
\begin{align}
\lambda_n &= \bl 0 \left(V\kl{n-1}-\sum_{m=1}^{n-1}\kl{n-m}\lambda_m\right), \label{energydef}\\
\kl n &= -\frac12\kl0\sum_{m=1}^{n-1}\bk m{n-m}+\frac{Q_{\lambda}}{\lambda_0-H_0}\left(V\kl{n-1}-\sum_{m=1}^{n-1}\kl{n-m}\lambda_m\right),\label{vectordef}
\end{align}
where the starting values $\lambda_0$, $\kl0$ are an eigenvalue and corresponding unit eigenvector of $H_0$, respectively, solve \cref{RSPTdef} while preserving the normalisation of $\sum_{n=0}^N\epsilon^n\kl n$ for every $N\ge0$ \cite[Ch.~11]{Cohen86}.
\end{theorem}

\subsection{Bloch sequences and diagrams}
\begin{figure}
\begin{center}
\begin{tikzpicture}[scale=0.5]
\draw 	(0,5.5) node[left=0.5cm]{$n=1$}
		(0,1) node[left=0.5cm]{$n=2$}
		
		(0,5) node[anchor=north east]{$(k_1)=$}\diagraml{{1}}
		(0,0) node[anchor=north east]{$(k_1,k_2)=$}\diagraml{{2,0}}
		(3,0)\diagraml{{1,1}}
		(6,0)\diagraml{{0,2}};
\drawgrid{(0,0)}{{2,0}}
		
\begin{scope}[xshift=13cm]
\draw[<->] (0,-0.5)--node[fill=white]{$n$}(6,-0.5);
\draw[<->] (6.5,0)--node[fill=white]{$n$}(6.5,6);
\draw[<->] (-0.5,0)--node[fill=white]{$k_1$}(-0.5,2);
\draw[<->] (0.5,2)--(0.5,3);
\draw (0.5,2.5)node[left]{$k_2$};

\draw (5,6)--(6,6)--(0,0)--(0,2)--(1,2)--(1,3)--(2,3);
\draw[dashed] (2,3)..controls (2,4)..(3,4)..controls (4,4) and (4,6)..(5,6);
\end{scope}
\end{tikzpicture}
\end{center}
\caption{Bloch diagrams. On the left are all diagrams for $n=1$ and $n=2$, including the non-convex $(0,2)$. On the right is some order $n$ diagram. Following Bloch~\cite{Bloch58}, we sometimes use curved, dashed lines to represent an arbitrary diagram.\label{diagrams_explained}}
\end{figure}
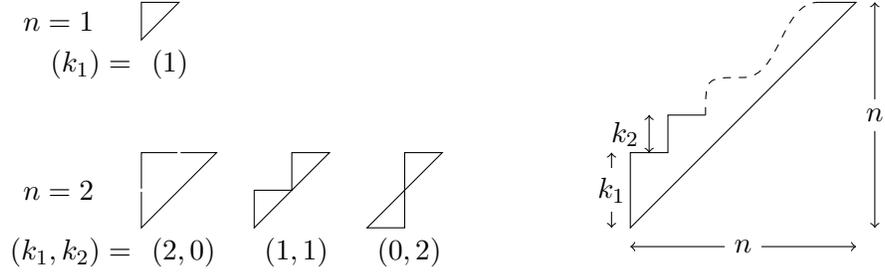
Bloch's seminal paper on degenerate RSPT~\cite{Bloch58} was the main inspiration for this paper.
In this subsection, we summarise the relevant definitions and results that we adopt from there.
These also apply to the nondegenerate case straightforwardly, see e.g.~\cite{Huby61}, and are adapted to our notation accordingly.
The result for the eigenvector is
\begin{equation}\label{vectorbloch}
\ket{\overline{\lambda_n}}=\sideset{}{'}\sum_{\{k\}_n} \svl k 1 2 n \kl0
\end{equation}
where the sum is over Bloch sequences of length $n$ defined by
\begin{equation}
{\{k\}_n}:\,\,\,\,\,\,k_i\in\mathbb{N}_0,\,i=1,2,\dots,n,\qquad\sum_{i=1}^n k_i=n,
\label{Blseq}
\end{equation}
and the prime indicates it is restricted to those sequences that satisfy
\begin{equation}\label{convex}
\sum_{i=1}^p k_i\ge p,\quad p=1,2,\dots,n-1.
\end{equation}
$\ket{\overline{\lambda_n}}$ is a solution to \cref{RSPTdef}; it is not normalised, but satisfies the condition $\left\langle\lambda_0\middle|\overline{\lambda_n}\right\rangle=0$, so that $\left\langle\lambda_0\middle|\overline{\lambda}\right\rangle=1$.
We distinguish it from the normalised one defined in \cref{RSPTRR} with an overline. 

Using \cref{vectorbloch}, \cref{P0comp} becomes
\begin{equation}\label{energybloch}
\lambda_n=\left\langle \lambda_0 \middle| V\middle| \overline{\lambda_{n-1}} \right\rangle=\sideset{}{'}\sum_{\{k\}_{n-1}} \bl0V\svl k 1 2 {n-1}\kl0.
\end{equation}

Bloch sequences can be represented graphically as staircase diagrams where step $i$ has height $k_i$ and width $1$, as illustrated in \cref{diagrams_explained}. The diagrams satisfying \cref{convex} are those that always stay above the diagonal. They are also called convex, and are known in the combinatorics literature as Dyck paths~\cite[p.76]{Flajolet09}.

\section{Stepwise diagrammatic solution\label{main}}
We find that when we require a normalised state vector, given the phase choice \cref{zeron}, the basic structure of the solution is retained:
\begin{theorem}\label{ceform}
The coupled recurrence relations in \cref{RSPTRR} are solved by $\lambda_n$, $\kl n$ of the form
\begin{align}
\lambda_n &= \sum_{\{k\}_{n-1}}e\!\left(k_1,k_2,\dots,k_{n-1}\right) \bl0 V\svl k 1 2 {n-1} \kl0, \label{edef}\\
\kl n &= \sum_{\{k\}_n} c\!\left(k_1,k_2,\dots,k_n\right) \svl k 1 2 n \kl0.\label{cdef}
\end{align}
where $c$, $e$ are rational-valued functions.
\end{theorem}
Note that the absence of the primes on the sums in \cref{edef,cdef} means that we must also allow non-convex diagrams (recall \cref{convex}).

Note further that we have introduced again a series for the eigenvalue \cref{edef}, which may appear unnecessary given the existing result of Bloch \cref{energybloch}, with known, simple values for the coefficients $e$. One should, however, view \cref{edef} as an auxiliary equation, not useful for actual evaluation of $\lambda_n$ (Bloch's result is best for that), but quite useful for obtaining the coefficients $c$ in \cref{cdef} according to the recurrence that is about to be derived. This is made possible by the non-uniqueness of the perturbation theory in terms of Bloch diagrams, which we explore more thoroughly in \cref{uniqueness}.

\begin{proof}
This is by complete induction on $n$.
For $n=1$, \cref{RSPTRR} gives the first-order corrections $\lambda_1=\bl0V\kl0$ and $\kl1=S^1V\kl0$.
These are of the form of \cref{edef,cdef} with $e(\emptyset)=c(1)=1$, proving the base case.

For $n>1$, assume \cref{edef,cdef} hold for all $m<n$.
We compute $\lambda_n$, $\kl n$ using \cref{RSPTRR}:
\begin{subequations}\label{eder}
\begin{align}
\lambda_n &=\bl0V\sum_{\{k\}_{n-1}}c\!\left(k_1,\dots,k_{n-1}\right)\sv k1{n-1}\kl0 \nonumber\\ &\phantom{=}- \bl0\sum_{m=1}^{n-2}\sum_{\{k\}_{n-m}}c\!\left(k_1,\dots,k_{n-m}\right)\sv k1{n-m}\kl0 \\ &\phantom{=}\times\sum_{\{j\}_{m-1}}e\!\left(j_1,\dots,j_{m-1}\right)\bl0V\sv j1{m-1}\kl0 \nonumber\\
&=\sum_{\{k\}_{n-1}}c\!\left(k_1,\dots,k_{n-1}\right)\bl0V\sv k1{n-1}\kl0 \nonumber\\ &\phantom{=}- \sum_{m=1}^{n-2}\sum_{\substack{\{k\}_{n-m}\\ k_1=0}}\sum_{\{j\}_{m-1}}c\!\left(0,k_2,\dots,k_{n-m}\right)e\!\left(j_1,\dots,j_{m-1}\right) \\ &\phantom{=}\times\bl0V\sv k2{n-m}S^0V\sv j1{m-1}\kl0,\nonumber
\end{align}
\end{subequations}
where in the second line we have changed the upper limit on the sum over $m$ from $n-1$ to $n-2$ since $\bk01=0$, and in the last line we used $\bl0 S^{k_1}=-\delta_{0,k_1}\bl0$ and $\kl0\bl0=-S^0$.
The result is of the form (\ref{edef}) with $e\!\left(k_1,\dots,k_{n-1}\right)$ given by \cref{eRR}.
In \cref{eRR}, the $\delta_{n-m,K_{n-m}}$ ensures that the arguments of $c$ and $e$ are Bloch sequences.

We repeat the same reasoning for the eigenvector
\begin{subequations}\label{cder}
\begin{align}
&\phantom{=}\kl n \nonumber\\&= -\frac12\kl0\sum_{m=1}^{n-1}\sum_{\{k\}_m}\sum_{\{j\}_{n-m}}\hspace{-2.66pt}c\!\left(k_1,\dots,k_m\right)c\!\left(j_1,\dots,j_{n-m}\right)\bl0\vs km1 \sv j1{n-m}\kl0\nonumber\\
&\phantom{=}+S^1V\sum_{\{k\}_{n-1}}c\!\left(k_1,\dots,k_{n-1}\right)\sv k1{n-1}\kl0 -S^1\sum_{m=1}^{n-1}\sum_{\{k\}_{n-m}}\sum_{\{j\}_{m-1}}c\!\left(k_1,\dots,k_{n-m}\right)\nonumber\\
&\phantom{=}\times e\!\left(j_1,\dots,j_{m-1}\right)\sv k1{n-m} \kl0\bl0V\sv j1{m-1}\kl0\\
&=\sum_{m=1}^{n-1}\sum_{\{k\}_m}\sum_{\{j\}_{n-m}}\frac12\left(1-\delta_{0,k_1}-\delta_{0,j_1}\right)c\!\left(k_1,\dots,k_m\right)c\!\left(j_1,\dots,j_{n-m}\right)\nonumber\\ &\phantom{=}\times S^0V\sv km2 S^{k_1+j_1}V\sv j2{n-m}\kl0 \nonumber\\ &\phantom{=}+\sum_{\{k\}_{n-1}}c\!\left(k_1,\dots,k_{n-1}\right)S^1V\sv k1{n-1}\kl0\nonumber\\
&\phantom{=}+\sum_{m=1}^{n-1}\sum_{\{k\}_{n-m}}\sum_{\{j\}_{m-1}}\left(1-\delta_{0,k_1}\right)c\!\left(k_1,\dots,k_{n-m}\right)e\!\left(j_1,\dots,j_{m-1}\right)\nonumber\\&\phantom{=}\times S^{k_1+1}V\sv k2{n-m}S^0V\sv j1{m-1}\kl0.
\end{align}
\end{subequations}
Here we have again used $\kl0\bl0=-S^0$ as well as
\begin{equation}
S^kS^j=\begin{cases}-S^0=-S^{k+j} & \text{if }k=j=0,\\ S^{k+j} & \text{if }k,j>0,\\0 & \text{else.}\end{cases}
\end{equation}
The result is of the form~(\ref{cdef}) with $c(k_1,\dots,k_n)$ given by \cref{cRR}.
Note that in \cref{eder,cder} (and therefore in equations throughout the following) argument lists can be empty.
Specifically, for $m=1$, $e(j_1,\dots,j_{m-1})=e(\emptyset)$.
This corresponds to an appearance of $\lambda_1$ in \cref{RSPTRR}.
\end{proof}
\begin{corollary}\label{ceRR}
The functions $c$ and $e$ defined in \cref{ceform} satisfy ($n\ge 2$)
\begin{align}
&\phantom{=}e\!\left(k_1,\dots,k_{n-1}\right)\nonumber\\ &= c\!\left(k_1,\dots,k_{n-1}\right) - \sum_{m=1}^{n-2} \delta_{0,k_{n-m}} \delta_{n-m,K_{n-m}} c\!\left(0,k_1,\dots,k_{n-m-1}\right) e\!\left(k_{n-m+1},\dots,k_{n-1}\right), \label{eRR}\\
&\phantom{=}c\!\left(k_1,\dots,k_n\right)\nonumber\\ &=\begin{dcases} \sideset{}{'}\sum_{m=1}^{n-1} \tfrac{1-\delta_{m,K_m}-\delta_{m,K_{m+1}}}2 c(m-K_m,\overbrace{k_m,\dots,k_2}^{\mbox{\rm\tiny{decreasing index}}}) c\!\left(K_{m+1}-m,k_{m+2},\dots,k_n\right), & k_1=0,\\
c\!\left(k_2,\dots,k_n\right), & k_1=1,\\
\sum_{m=1}^{n-1} \delta_{0,k_{n-m+1}}\delta_{n-m,K_{n-m}-1} c\!\left(k_1-1,k_2,\dots,k_{n-m}\right) e\!\left(k_{n-m+2},\dots,k_n\right), & k_1>1,\end{dcases}\label{cRR}
\end{align}
with the primed sum restricted to $k_{m+1}\ge m-K_m \ge 0$ so that all the arguments of $c$ are non-negative, and $K_m=\sum_{i=1}^m k_i$.
For $m=1$, in some of the argument lists, the initial index is smaller than the final index; as in \cref{nullsum} such an argument list should be interpreted as an empty set.
The starting values of $c$ and $e$ ($n=1$) are
\begin{equation}
c(1)=1,\qquad e(\emptyset)=1.\label{ce0}
\end{equation}
\end{corollary}
Note that functions $c$ and $e$ are independent of the Hamiltonian.
We will refer to the three cases in \cref{cRR} as the $k_1=0$, $k_1=1$, and $k_1>1$ rules.
For a diagrammatic explanation of the recurrence relations, refer to \cref{dRR}.

In \cref{vectorbloch,energybloch}, all diagrams are summed up with a coefficient of $1$, or $0$ if the sum is extended to non-convex diagrams.
The same cannot be true for $c$ and $e$ because of the factor $1/2$ in $\bk0n$.
A nonzero $\bk0n$ means that some diagrams start below the diagonal, so are definitely not convex.
And the factor $1/2$ means their coefficient is generally unequal to $1$.

From calculating $c$ for all diagrams up to fourth order, cf. \cref{example}, and selected higher-order diagrams, we anticipate that it will have the following property, which will be useful in the subsequent analysis:
\begin{definition}[Crossing Property]
For any Bloch sequence $(k_1,\dots,k_n)$ let $x(k_1,\dots,k_n)$ be the number of times its associated diagram intersects the main diagonal. We say a function $f$ has the crossing property if there is another function $g$ such that
\begin{equation}
f(k_1,\dots,k_n)=g(\ceil{x(k_1,\dots,k_n)/2})
\end{equation}
for all Bloch sequences $(k_1,\dots,k_n)$, i.e. $f$ depends only on the number of times a diagram crosses \emph{from below to above} the main diagonal. Here $\ceil{\cdot}$ is the ceiling function.   
\end{definition}
If the function $c$ has the crossing property, the problem of evaluating it only needs to be performed on a set of representative diagrams. These diagrams can be taken to be the ones with the Bloch sequences (cf. Eq.~(\ref{Blseq})) 
\begin{equation}
\{k\}_{2n}=(0,2)^n,
\end{equation}
meaning 0,2 repeated $n$ times.
It is helpful below to have a separate symbol for these specific instances of the $c$ function: 
\begin{definition}\label{deft}
$t(n)=c\!\left((0,2)^n\right)$ for $n>0$, and $t(0)=c(1)=1$.
\end{definition}
$t(n)$ can be computed:
\begin{lemma}\label{tdefl}
\begin{equation}\label{tdef}
t(n)=\binom{n-\frac12}n=\frac1{2^{2n}} \binom{2n}n=\frac{\Gamma\!\left(n+\frac12\right)}{\sqrt\pi\Gamma(n+1)}.
\end{equation}
\end{lemma} 
\begin{proof}
The generalised binomial coefficient is defined in the usual way
\begin{equation}
\binom rn=\frac{r(r-1)\dots(r-n+1)}{n!}.
\end{equation}

We use the $k_1=0$ and $k_1=1$ rules of \cref{cRR} to derive a recurrence relation for $t(n)$,
\begin{align}\label{tRR}
t(n) &\overset{k_1=0}{=} \frac12\left[ t(0)c\!\left(1,(0,2)^{n-1}\right)-t(1)t(n-1)+c(1,0,2)c\!\left(1,(0,2)^{n-2}\right)-t(2)t(n-2)+\dots\right.\nonumber\\
&\phantom{\overset{k_1=0}{=}}\left.+c\!\left(1,(0,2)^{n-1}\right)c(1)\right]\nonumber\\
&\overset{k_1=1}{=}\frac12\left[t(0)t(n-1)-t(1)t(n-1)+t(1)t(n-2)-t(2)t(n-2)+\dots+t(n-1)t(0)\right]\nonumber\\
&\overset{\phantom{k_1=0}}{=}t(0)t(n-1)-t(1)t(n-1)+t(1)t(n-2)-t(2)t(n-2)+\dots+\frac{(-1)^{n-1}}2t\!\left(\floor*{\frac n 2}\right)^2\nonumber\\
&\overset{\phantom{k_1=0}}{=}\sum_{m=0}^{n-1}(-1)^m t\!\left(\ceil*{\frac m2}\right)t\!\left(n-1-\floor*{\frac m2}\right)2^{-\delta_{m,n-1}}\\
&\overset{\phantom{k_1=0}}{=}\frac12\sum_{m=0}^{2(n-1)}(-1)^m t\!\left(\ceil*{\frac m2}\right)t\!\left(n-1-\floor*{\frac m2}\right)\nonumber\\
&\overset{\phantom{k_1=0}}{=}\frac12\sum_{m=0}^{n-1}t(m)t(n-1-m)-\frac12\sum_{m=1}^{n-1}t(m)t(n-m).\nonumber
\end{align}
Here $\floor{\cdot}$, $\ceil{\cdot}$ are floor and ceiling function, respectively, rounding to the nearest integer lesser/greater than the argument.
The summands are symmetric under reversing the order of summation, so for all but one term the factor $1/2$ cancels.
But we find it convenient to instead keep all terms and group them by even and odd indices, as done in the last line of \cref{tRR}.
Then by bringing the latter sum to the left hand side, which we can also write as $t(n)=\frac12t(n)t(0)+\frac12t(0)t(n)$,  and multiplying by $2$, we can rewrite \cref{tRR} as
\begin{equation}\label{tconv}
\sum_{m=0}^n t(m)t(n-m)=\sum_{m=0}^{n-1} t(m)t(n-1-m)=t(0)^2=1,
\end{equation}
i.e. we find that the sum is independent of $n$, so we can set e.g. $n=1$ to evaluate it.

To complete the proof constructively,\footnote{Alternatively, we could now confirm that \cref{tdef} satisfies \cref{tconv}, with the uniqueness of the solution being guaranteed from the recursive construction of $c$.} consider that \cref{tconv} has the form of a discrete convolution, so we can restate it in terms of the (ordinary) generating function of $t$,
\begin{equation}
g(x)=\sum_{n=0}^\infty t(n)x^n,
\end{equation}
as
\begin{equation}
g(x)^2=\sum_{n=0}^\infty x^n=\frac1{1-x}
\end{equation}
with a geometric series, so we can express the generating function as a binomial series to determine $t$,
\begin{equation}
g(x)=\frac1{\sqrt{1-x}}=\sum_{n=0}^\infty\binom{-\frac12\,}n(-x)^n=\sum_{n=0}^\infty \binom{n-\frac12}nx^n
\end{equation}
which gives \cref{tdef}.
\end{proof}
Note that $t(n)$ decreases only slowly with $n$; asymptotically, $t(n)\sim 1/\sqrt{\pi n}$.

The solution for $e$ is slightly more complicated as, in contrast to $c$, it does not depend solely on the diagram's topology w.r.t. the main diagonal, but also w.r.t. the \emph{upper diagonal}, which is defined as the diagonal line one unit higher than the main diagonal, as illustrated in Fig.~\ref{cross}.
\begin{definition}[Crossing Numbers]
We say a Bloch sequence $(k_1,\dots,k_n)$ has crossing numbers $N_1,n_1,N_2,n_2,\dots,N_m,n_m$ if its associated diagram crosses, in order, above the upper diagonal $N_1$ times, below the main diagonal $n_1$ times, then above the upper diagonal $N_2$ times, etc.
Here $m$ is some integer with $1\le m\le n/3+1$.
For concreteness there is always an even number of crossing numbers, where the first and last one, $N_1$ and $n_m$, may be $0$ while all other ones are strictly positive integers such that $m$ is well-defined.

Given a Bloch sequence $(k_1,\dots,k_n)$, the crossing numbers can be constructed as follows:
\begin{equation}\label{crossalg}
\hbox{\begin{lstlisting}[mathescape=true,language=C]
$m=1$; $N_m=0$; $n_m=0$; $x=1$;
for $i=1$, $i\le n$, $i++$
    if $\displaystyle\sum_{j=1}^ik_j>i$ $\land$ $\displaystyle\sum_{j=1}^{i-1}k_j\le i-1$
        if $x\ne1$
            $m++$; $N_m=0$; $n_m=0$; $x=1$;
        $N_m++$;
    else if $\displaystyle\sum_{j=1}^ik_j<i$ $\land$ $\displaystyle\sum_{j=1}^{i-1}k_j\ge i-1$
        $n_m++$; $x=0$;
\end{lstlisting}}
\end{equation}
\end{definition}
A canonical diagram that has crossing numbers $N_1,\dots,n_m$ is \begin{equation}
(k_1,\dots,k_n)=((2,0)^{N_1},(0,2)^{n_1-1},0,3,0,(2,0)^{N_2-1},(0,2)^{n_2-1},0,3,0,\dots,(2,0)^{N_m-1},(0,2)^{n_m})
\end{equation}
some examples of which are given in \cref{cross}.
For the special case that the crossing numbers are $N_1,n_1=0,0$, this sequence is empty, and we can instead take $(k_1)=(1)$ as this canonical diagram.
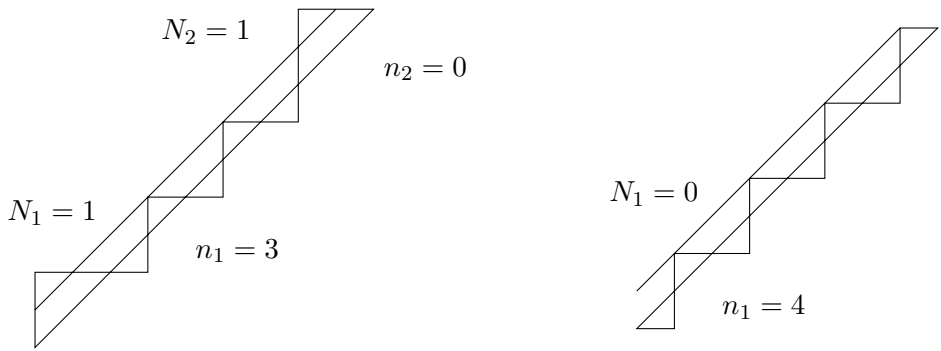
\begin{figure}
\begin{tikzpicture}[scale=0.5]
\draw (0,0) \diagram{{2,0,0,2,0,2,0,3,0}} (0,1)--(8,9);
\draw	(-1,3) node[anchor=south west]{$N_1=1$}
		(6,9) node[anchor=north east]{$N_2=1$}
		(4,2) node[anchor=south west]{$n_1=3$}
		(9,8) node[anchor=north west]{$n_2=0$};
\draw[xshift=16cm,yshift=0.5cm] (0,0) \diagram{{0,2,0,2,0,2,0,2}} (0,1)--(7,8) (-1,3) node[anchor=south west]{$N_1=0$} (2,0) node[anchor=south west]{$n_1=4$};
\end{tikzpicture}
\caption{Illustration of crossing numbers. The diagram $(2,0,0,2,0,2,0,3,0)$ is the lowest order diagram with crossing numbers $1,3,1,0$. The diagram $(0,2)^n$ from \cref{deft} is the lowest order diagram with crossing numbers $0,n$, shown here with $n=4$. Both the (main) diagonal and the upper diagonal are illustrated here. \label{cross}}
\end{figure}

The upper limit $m\le n/3+1$ is derived by setting $N_1=n_m=0$, and all other $N_i=n_i=1$ in the lowest order diagram.
Such a large $m$ is somewhat of an outlier though.
If we consider a Bloch diagram (rotated by $-\pi/4$) as a random bridge, we could develop the notion of a \emph{typical} diagram.
The number of times an $n$th order diagram touches or intersects the diagonal, which is an upper bound on $m$, asymptotically follows a Rayleigh distribution with mean $\sqrt{\pi n}$ \cite[p.708]{Flajolet09}.
This seems to indicate that typically $m\sim\sqrt n$, i.e. in most instances $m\ll n$.

We will now proceed to the main result of the paper, \cref{mainthm}, in which explicit formulas for $c$ and $e$ are obtained.
We first briefly review the heuristics that led us to the formulation of this theorem.
We noted that if we assumed that $c$ had the crossing property, \cref{tdefl} would be sufficient to calculate $c$ for all diagrams.
By calculating a number of examples, we made observations about the structure of the solution, noting the dependence on the crossing numbers only, and used these to allow further simplification of the recurrence relations.
We came to an ansatz for solving the coupled recurrence relations, guided by the observation that our solution for $e$ has to be consistent with $c$ having the crossing property.

In the end, the ansatz is proved in the following theorem by induction, accompanied by a straightforward algebraic analysis:

\begin{theorem}\label{mainthm}
Let $(k_1,\dots,k_n)$ be a Bloch sequence with crossing numbers $N_1,n_1,\dots,N_m,n_m$.
The functions $c$ and $e$ defined in \cref{ceform} are
\begin{align}
c\!\left(k_1,\dots,k_n\right)&=t\!\left(\sum_{i=1}^m n_i\right),\label{cres}\\
e\!\left(k_1,\dots,k_n\right)&=\sum_{i=1}^m \left[t\!\left(\sum_{l=1}^i N_l\right)-t\!\left(\sum_{l=1}^{i-1} N_l\right)+\delta_{i,1}\right]t\!\left(\sum_{j=i}^m n_j\right),\label{eres}
\end{align}
with $t(x)=\binom{2x}x 2^{-2x}$ as given in \cref{tdef}, i.e. $c$ has the crossing property and $e$ is a function of the crossing numbers only.
\end{theorem}
\begin{proof}
We verify that \cref{cres,eres} are consistent with \cref{ceRR}.
At $n=1$, we have $k_1=1$ with $N_1=n_1=0$, which are also the crossing numbers for an empty diagram ($\emptyset$, $n=0$).
Then \cref{cres,eres} give $c(1)=e(1)=e(\emptyset)=1$, consistent with \cref{ce0}.

Suppose \cref{mainthm} holds for all diagrams of degree less than $n$. (Note that degree simply refers to the number of entries in the Bloch sequence $\{k\}_n$.)
We apply \cref{eRR} to compute $e(k_1,\dots,k_n)$ and show it is consistent with \cref{eres}:
\begin{subequations}\label{sproof}
\begin{align}
e\!\left(k_1,\dots,k_n\right)&=t\!\left(\sum_{j=1}^m n_j\right)-\sum_{i=1}^m\sum_{k=1}^{N_i}t\!\left(\sum_{j=1}^{i-1}N_j+k\right)\nonumber\\&\phantom{=}\times\sum_{l=i}^m\left[t\!\left(\sum_{h=i}^l N_h-k\right)-\left(1-\delta_{l,i}\right)t\!\left(\sum_{h=i}^{l-1} N_h-k\right)\right]t\!\left(\sum_{g=l}^m n_g\right)\\
&=\label{s1}t\!\left(b_1\right)-\sum_{i=1}^m\sum_{k=1+a_{i-1}}^{a_i}t(k)\sum_{l=i}^m \left[t\!\left(a_l-k\right)-\left(1-\delta_{l,i}\right)t\!\left(a_{l-1}-k\right)\right]t\!\left(b_l\right)\\
&=\label{s2}t\!\left(b_1\right)-\sum_{l=1}^m t\!\left(b_l\right)\sum_{i=1}^l\sum_{k=1+a_{i-1}}^{a_i}t(k) \left[t\!\left(a_l-k\right)-\left(1-\delta_{l,i}\right)t\!\left(a_{l-1}-k\right)\right]\\
&=\label{s3}t\!\left(b_1\right)-\sum_{l=1}^m t\!\left(b_l\right) \left[\sum_{k=1}^{a_l} t(k)t\!\left(a_l-k\right) - \sum_{k=1}^{a_{l-1}}t(k)t\!\left(a_{l-1}-k\right)\right]\\
&=\label{s4}\sum_{l=1}^m t\!\left(b_l\right)\left[t\!\left(a_l\right)-t\!\left(a_{l-1}\right)+\delta_{l,1}\right]
\end{align}
\end{subequations}
where in \cref{s1} we introduce $a_i=\sum_{j=1}^i N_j$ and $b_i=\sum_{j=i}^m n_j$ to simplify notation, and shift the summation index $k$ by $a_{i-1}$.
Then in \cref{s2} we switch the sums over $i$ and $l$.
Note that $a_0=0$, $a_1=N_1\ge0$, and $a_{i+1}>a_i$ for $i>0$.
So in \cref{s3} we can combine the double sum over $i,k$ into one over $k$.
And finally in \cref{s4} we add and subtract the $k=0$ terms, then use \cref{tconv} and get \cref{eres}.

Similarly, we calculate $c(k_1,\dots,k_n)$ using \cref{cRR}
\begin{subequations}
\begin{align}
&\phantom{=}c\!\left(k_1,\dots,k_n\right)\nonumber\\&=\delta_{0,k_1}\frac12 \left[\sum_{i=1}^{b_1}t(i-1)t\!\left(b_1-i\right)-\sum_{i=1}^{b_1-1} t(i)t\!\left(b_1-i\right)\right]+\delta_{1,k_1}t\!\left(b_1\right)\label{c1}\\
&\phantom{=}+\left(1-\delta_{0,k_1}-\delta_{1,k_1}\right)\sum_{i=1}^m\sum_{j=1}^{N_i}t\!\left(a_{i-1}+j-1\right)s\!\left(N_i-j,n_i,\dots,N_m,n_m\right)\nonumber\\
&=\delta_{0,k_1}\frac12 \left[\sum_{i=0}^{b_1-1}t(i)t\!\left(b_1-1-i\right)-\left(1-2t\!\left(b_1\right)\right)\right]+\delta_{1,k_1}t\!\left(b_1\right)\label{c2}
+\left(1-\delta_{0,k_1}-\delta_{1,k_1}\right)\\&\phantom{=}\times\sum_{i=1}^m\sum_{j=1}^{N_i}t\!\left(a_{i-1}+j-1\right)\sum_{k=i}^m t\!\left(b_k\right)\left[t\!\left(\sum_{h=i}^k N_h-j\right)-\left(1-\delta_{k,i}\right)t\!\left(\sum_{h=i}^{k-1} N_h-j\right)\right]\nonumber\\
&=\left(\delta_{0,k_1}+\delta_{1,k_1}\right)t\!\left(b_1\right)\label{c3}\\ &\phantom{=}+\left(1-\delta_{0,k_1}-\delta_{1,k_1}\right)\sum_{k=1}^m t\!\left(b_k\right)\sum_{i=1}^k\sum_{j=a_{i-1}}^{a_i-1}t(j)\left[t\!\left(a_k-1-j\right)-\left(1-\delta_{k,i}\right)t\!\left(a_{k-1}-1-j\right)\right]\nonumber\\
&=\left(\delta_{0,k_1}+\delta_{1,k_1}\right)t\!\left(b_1\right)\label{c4}\\ &\phantom{=}+\left(1-\delta_{0,k_1}-\delta_{1,k_1}\right)\sum_{k=1}^m t\!\left(b_k\right)\left[\sum_{j=0}^{a_k-1} t(j)t\!\left(a_k-1-j\right)-\sum_{j=0}^{a_{k-1}-1}t(j)t\!\left(a_{k-1}-1-j\right)\right]\nonumber\\
&=\left(\delta_{0,k_1}+\delta_{1,k_1}\right)t\!\left(b_1\right)\label{c5}+\left(1-\delta_{0,k_1}-\delta_{1,k_1}\right)\sum_{k=1}^m t\!\left(b_k\right)\left[1-\left(1-\delta_{k,1}\right)\right]\\
&=t\!\left(b_1\right).
\end{align}
\end{subequations}
In \cref{c1} note that for $k_1=0$ the diagram will have $2b_1-1$ intersections with the main diagonal, $b_1-1$ of which are horizontal and thus come with a negative sign. (This rule is a consequence of the negative sign in Eq.~(\ref{cRR}), see Sec.~\ref{dRR}, Fig.~\ref{dRRe}.)
Then in \cref{c2} we shift the index of the first sum, as well as add and subtract the $i=0,b_1$ terms to the second sum and immediately evaluate it with \cref{tconv}, which we then also apply to the first sum in the following step.
In \cref{c3} we index-shift the $j$ sum by $a_{i-1}-1$, and switch the $k$ and $i$ sums.
In \cref{c5} we use \cref{tconv} again, but we have to be careful not to apply it if the sums vanish because of \cref{nullsum}.
Since we are in the $k_1>1$ term, we know that $a_1=N_1\ge1$, and $a_{i+1}>a_i$ for $i\ge1$ still holds, so the only term \cref{nullsum} applies to is the second $j$ sum for $k=1$ since $a_0=0$.
\end{proof}

\section{On non-uniqueness of diagrammatic representations\label{uniqueness}}
By stating a recurrence relation and initial conditions we uniquely define a quantity.
For example, combined with the initial conditions, \cref{energydef,vectordef} fix the eigenenergy and eigenvector corrections, and \cref{eRR,cRR} uniquely define the coefficients $c$ and $e$.
That does not mean, however, that $c$ and $e$ are necessarily the unique solutions of \cref{energydef,vectordef}, or that \cref{energydef,vectordef} are the unique solutions of \cref{RSPTdef}.

Since we are considering the nondegenerate case, eigenvectors are determined up to a factor.
We are fixing the normalisation with \cref{norm}, but that still leaves a phase freedom.
The zeroth order phase is set by our choice of $\kl0$.
We can modify this phase in higher orders of $\epsilon$ by adding an imaginary part to \cref{zeron}.
An arbitrary imaginary part would generally change the structure of \cref{cdef}, but we could preserve it, e.g. by setting
\begin{equation}\label{altphase}
\bk0n\rightarrow\begin{cases} -\sum_{m=1}^{(n-1)/2}\bk m{n-m} & \text{if }n\text{ is odd}, \\
-\sum_{m=1}^{n/2-1} \bk m{n-m} -\frac12 \bk{n/2}{n/2} & \text{if }n\text{ is even}.\end{cases}
\end{equation}
This reduces the number of diagrams but comes at the cost of a more complicated rule requiring an even/odd distinction.

Note that if the Hamiltonian is real-symmetric, \cref{zeron,altphase} are equivalent, equal, yet \cref{altphase} still provides the more compact description in terms of the number of diagrams.
This brings us to the main point of this section: Once norm and phase are fixed, the eigenvector is uniquely defined, but the representation in terms of diagrams is not.
The eigenenergy is of course independent of the factor in front of the eigenvector but has a similar freedom with regard to the decomposition into diagrams.

\begin{definition}
A Bloch sequence $(k_1,\dots,k_n)$ containing $q-1$ zeroes, $q=1,\dots,n$, can be represented equivalently by $q$ strings of positive integers $z_i\in\mathbb{N}^k$, $k=0,\dots,n-q+1$ (note that null strings, $k=0$, are allowed).
$\mathcal{Z}$ is defined as the mapping between a Bloch sequence and the set of $z_i$ strings:
\begin{equation}
\mathcal{Z}:\left(k_1,\dots,k_n\right)\mapsto\left(z_1,\dots,z_q\right)\quad:\quad\left(k_1,\dots,k_n\right)=\left(z_1,0,z_2,\dots,0,z_q\right).
\end{equation}
We also define the operators $T$, $L$, and $D$ (``total", ``length", and ``difference") applied to string $z$:
\begin{equation}
\text{for }z=j_1j_2\dots j_k:\quad T(z)=\sum_{i=1}^k j_i,\quad L(z)=k,\quad D(z)=T(z)-L(z).
\end{equation}
\end{definition}
As an example for the map $\mathcal{Z}$, we can write
\begin{equation}
\mathcal{Z}:\left(1,3,0,0,1\right)\mapsto\left(\{1,3\},\O,\{1\}\right).   
\end{equation}
Here $n=5$, $q=3$, and we see the appearance of integer strings $z_i$ of varying length, including the null string.

\begin{theorem}\label{permutations}
In the expansion of the eigenenergy, \cref{edef}, all diagrams that differ only by a permutation of z strings correspond to the same matrix element.

Similarly, in the expansion of the eigenvector, \cref{cdef}, diagrams that share the same first string and otherwise differ only by a permutation of the remaining strings correspond to the same matrix element.
\label{th4}
\end{theorem}
\begin{proof}
Suppose the $m$th component of a Bloch sequence $(k_1,\dots,k_n)$ vanishes, $k_m=0$.
The term this sequence contributes to the energy correction is
\begin{align}
\begin{split}\label{permenergy}
&\bl0V\sv k1{m-1}S^0V\sv k{m+1}n\kl0\\
&=-\bl0V\sv k1{m-1}\kl0\bl0V\sv k{m+1}n\kl0\\
&=-\bl0V\sv k{m+1}n\kl0\bl0V\sv k1{m-1}\kl0\\
&=\bl0V\sv k{m+1}nS^0V\sv k1{m-1}\kl0.
\end{split}
\end{align}
Suppose $k_m$ is the $M$th zero in the Bloch sequence, and $\mathcal{Z}(k_1,\dots,k_n)=(z_1,\dots,z_q)$.
\Cref{permenergy} implies that $(z_1,\dots,z_M,z_{M+1},\dots,z_q)$ has the same operator content as $(z_{M+1},\dots,z_q,z_1,\dots,z_M)$, i.e. the operator content is invariant under cyclical permutation of strings.
Now, let $k_j$ be the $J$th zero, w.l.o.g. assume $j>m$
\begin{align}
&\bl0V\sv k1{m-1}S^0V\sv k{m+1}{j-1}S^0V\sv k{j+1}n\kl0\nonumber\\
&=(-1)^2\bl0V\sv k1{m-1}\kl0\bl0V\sv k{m+1}{j-1}\kl0\bl0V\sv k{j+1}n\kl0\nonumber\\
&=(-1)^2\bl0V\sv k{m+1}{j-1}\kl0\bl0V\sv k1{m-1}\kl0\bl0V\sv k{j+1}n\kl0\nonumber\\
&=\bl0V\sv k{m+1}{j-1}S^0V\sv k1{m-1}S^0V\sv k{j+1}n\kl0,
\end{align}
i.e. $(z_1,\dots,z_q)$ has the same operator content as $(z_{M+1},\dots,z_J,z_1,\dots,z_M,z_{J+1},\dots,z_q)$.
For example, by setting $M=1$ we can permute the first string $z_1$ to the $J$th position without changing the order of the other strings.
From this we can compose all permutations.

For the eigenvector, the calculation works out analogously with the only difference that here the operator content does not start with a $\bl0V$, so we can never permute $z_1$.
The rest of the diagram $(z_2,\dots,z_q)$ has the same structure\footnote{Though remember that unless $z_1=1$, it is no longer equivalent to a Bloch sequence, so not a diagram by itself, which does not impact the permutation of matrix elements of course.} as a term in the energy expansion, so the same permutation rules apply.
\end{proof}

Part of the redundancy identified by Theorem \ref{th4} already appears when we define the recurrence relations for $c$ and $e$.
For example, in the last term of \cref{energydef}, we can switch the order to $\lambda_m\bk0{n-m}$, which would change \cref{eder} and would lead to a different recurrence relation for $e$ and thus different values for $c$ and $e$.

We could also compare to the Bloch style result for the energy, \cref{energybloch}, which is a sum over all convex diagrams, and note that in our language, a convex diagram is described as $N_1,0$ with $e=t(N_1)>0$, but there are also many non-convex diagrams for which $e\ne0$.
So our result for the energy is less efficient.
But even when restricting to convex diagrams, \cref{permutations} still leads to a lot of redundancy.
Salzman \cite{Salzman68} addresses this for the (unnormalised) eigenvector by separating terms into an operator part (what we would call $z_1$) and a coefficient containing the matrix elements.
The number of different $z_1$'s in an order $n$ convex diagram is $2^n-n$. (Salzman already gave this as a sum, we just confirmed and evaluated it.)
Unfortunately, the rules he gives to list all diagrams are relatively complicated and equivalent coefficients are collected manually.
Silverstone and Holloway \cite{Silverstone70}, again for unnormalised eigenvectors, give a formally minimal result which still requires evaluating many derivatives.

To reduce the number of diagrams in our result down to a minimum, we can sum up $c$ and $e$ for all diagrams that are equivalent by \cref{permutations} and only keep one representative diagram.
For example, we can declare an ordering on strings and choose as representative diagram the one where strings are ordered descending.
\begin{definition}[Ordering of strings]
Let $y=(k_1,\dots,k_n)\ne z=(j_1,\dots,j_m)$ be strings of positive integers. We say $y>z$
\begin{gather}
\text{if }D(y)>D(z),\nonumber\\
\text{else if }L(y)>L(z),\nonumber\\
\text{else if }k_1>j_1,\\
\text{else if }k_2>j_2,\nonumber\\
\vdots\nonumber
\end{gather}

The canonical representative of a permutation group of strings has $z_1\ge z_2\ge\dots\ge z_q$.
\end{definition}
Giving the sum over all $c$ or $e$ for an arbitrary representative diagram is generally not an easy task.
Of course, given a string representation $z_1^{m_1}\dots z_k^{m_k}$ ($k$ distinct strings $z_i$ with multiplicity $m_i$), we can write down all the $(\sum_{i=1}^km_i)!/\prod_{i=1}^km_i!$ permutations, calculate their $c$ and $e$ and sum them up to get a $c_\text{eff}$ and $e_\text{eff}$.
The difficulty lies in automating this, i.e. listing only the canonical diagrams and finding an explicit function on them that gives $c_\text{eff}$ and $e_\text{eff}$ directly, preferably without having to invoke \cref{crossalg} for the whole permutation group.
This is less of a concern for the energy where we can alternatively start from \cref{energybloch}.
Then the problem becomes counting all the convex permutations, a nested sum for which can be written down but perhaps cannot be evaluated explicitly without specifying a Bloch sequence first.

\subsection{Number of terms}
We take a look at how many diagrams are generated by our method and other previous methods, and how many of them may correspond to distinct operator expressions.
This subsection is summarised in \cref{termcounts}.

At order $n$, there are $\binom{2n-1}{n}$ distinct Bloch sequences~\cite{Bloch58}.
This can easily be seen by considering that to construct all diagrams we have to list all distinct arrangements of $n$ unit vertical steps and $n-1$ unit horizontal steps (the $n$-th horizontal step is always fixed at the end).
This is the number of terms in our perturbation expansion for $\kl n$ and $\lambda_{n+1}$ (though $e$ can be $0$).
If we apply \cref{permutations}, it becomes an upper bound for the number of canonically ordered diagrams, i.e. the minimum number of terms required to cover all distinct operators.
Asymptotically it scales as $4^n/2\sqrt{\pi n}$.

From Bloch~\cite{Bloch58} we know that convex diagrams are sufficient for the expansion of the energy (or the unnormalised vector). The number of these diagrams for order $n$ is simply the Catalan numbers $C_n=(2n)!/n!(n+1)!=\frac2{n+1}\binom{2n-1}{n}$~\cite{Bloch58};\cite[p.76]{Flajolet09}.
Asymptotically $C_n\sim4^n/\sqrt{\pi n^3}$ \cite[p.7]{Flajolet09}, i.e. the exponential scaling is the same, only the algebraic pre-factor is improved.

A lower bound for the minimum number of diagrams is the number of partitions of $n$ into positive integers, cf.~\cite{Silverstone70}.
There is no known explicit expression for this partition function, but it has a generating function, recurrence relations, and an asymptotic expression $\exp(\pi\sqrt{ 2n/3})/4\sqrt3n=4^{\pi\ln{4}\sqrt{2n/3}}/4\sqrt3n$ \cite[p.41]{Flajolet09}.

As stated above, Salzman~\cite{Salzman68} grouped diagrams by $z_1$ (the string of positive integers before the first $0$ in the Bloch sequence) and counted $2^n-n$ distinct groups within convex diagrams of length $n$.
We can view this as a lower bound on the number of terms in the unnormalised eigenvector correction $\left|\overline{\lambda_n}\right\rangle$, since $z_1$ cannot be permuted with the other strings without changing the operator content.
By adding the number of $z_1$'s leading to a non-convex diagram, we can generalise this to a lower bound for the number of terms in $\kl n$: $2^n-1$.
For the energy the situation is slightly more complicated as diagrams with a distinct $z_1$ can still be equivalent.
For convex diagram this becomes relevant at $n\ge5$, which is why it does not appear in Figs.~\ref{diagrams123} or \ref{diagrams4}, e.g. $(3,0,2,0,0)$ is equivalent to $(2,0,3,0,0)$.
Yet, any of the $2^n-n$ $z_1$'s that can start a convex diagram can be the greatest string of a canonically ordered diagram, so this lower bound also applies to $\lambda_{n+1}$ after all.

\begin{table}
\begin{tabular}{|m{1cm}||m{1.5cm}|m{1.5cm}|m{1.5cm}|m{1.5cm}|m{0.4cm}m{1.1cm}|m{0.4cm}m{1.1cm}|m{1.5cm}|m{1.5cm}||m{1.5cm}|}
\hline
order & here/all & \multicolumn{2}{l|}{Kato~\cite{Kato49}} & Bloch~\cite{Bloch58} & \multicolumn{4}{l|}{minimum} & \multicolumn{2}{l||}{min. if $V$ off-diag.} & partition \\
$n$&Bl. seq.& $P_n$ & $\lambda_{n+1}$ &(convex)& \multicolumn{2}{l|}{$\kl n$} & \multicolumn{2}{l|}{$\lambda_{n+1}${\textsuperscript{a}}} & $\kl n$ & $\lambda_{n+1}${\textsuperscript{a}} &function\\
\hline\hline
$1$ & $1$ &$2$&$3$& $1$ &$1$&$(\ge1)$&$1$&$(\ge1)$ &$1$&$1$ & $1$\\
\hline
$2$ & $3$ &$6$&$10$& $2$ &$3$&$(\ge3)$&$2$&$(\ge2)$ &$2$&$1$& $2$\\
\hline
$3$ & $10$ &$20$&$35$& $5$ &$9$&$(\ge7)$&$5$&$(\ge5)$ &$5$&$2$& $3$\\
\hline
$4$ & $35$ &$70$&$126$& $14$ &$26$&$(\ge15)$&$13$&$(\ge12)$ &$12$&$4$& $5$\\
\hline
$n$ & $\binom{2n-1}n$ &$\binom{2n}n$& $\binom{2n+1}{n}$ & $\frac{(2n)!}{n!(n+1)!}$ &\multicolumn{2}{l|}{$\ge2^n-1$}&\multicolumn{2}{l|}{$\ge2^n-n$}&&&recursive\\
\hline
asym. & $\frac{4^n}{2\sqrt{\pi n}}$ &$\frac{4^n}{\sqrt{\pi n}}$&$\frac{4^n2}{\sqrt{\pi n}}$& $\frac{4^n}{\sqrt{\pi n^3}}$ &&&&&&&$\frac{e^{\pi\sqrt{ 2n/3}}}{4\sqrt3n}$\\
\hline
\multicolumn{12}{l}{{\textsuperscript{a}}\footnotesize{ Up to fourth order, $\left|\overline{\lambda_n}\right\rangle$ has the same number of terms, but in higher orders it has more than $\lambda_{n+1}$.}}
\end{tabular}
\caption{Number of terms in the energy correction $\lambda_{n+1}$, the unnormalised eigenvector correction $\left|\overline{\lambda_n}\right\rangle$, the normalised eigenvector correction $\kl n$, and the correction to the projector onto a degenerate eigenspace $P_n$, by order $n$, including the asymptotic behaviour in the last row, where known. The last column shows the number of partitions of $n$ into positive integers, a lower bound for the minimum number of terms we would get when summarising terms according to \cref{permutations}. These minimum numbers are also shown in columns $6$ through $9$, with Salzman-style~\cite{Salzman68} lower bounds in lieu of an explicit expression for general $n$. Where not otherwise specified, the number of terms in the energy and vector corrections is the same, i.e. in column $2$, the number of all Bloch sequences is the number of terms in $\kl n$ and $\lambda_{n+1}$, and in column $5$, the number of convex diagrams gives the number of terms in $\left|\overline{\lambda_n}\right\rangle$ and $\lambda_{n+1}$ following Bloch~\cite{Bloch58}. If $V$ is purely off-diagonal (columns $8$ and $9$), none of the lower bounds apply.\label{termcounts}}
\end{table}
Clearly none of the bounds are tight for sufficiently large $n$, though the latter set of lower bounds show that the minimal number of diagrams scales as $\exp(c n)$ rather than $\exp(c \sqrt n)$.

All these considerations remain independent of the Hamiltonian.
If we include information about the Hamiltonian, further simplifications can be made.
For example, as noted above, if the Hamiltonian is real-symmetric the operator content of a diagram is invariant under reversing the ordering within strings (except for $z_1$ in the eigenvector expansion).
A more generally applicable scenario is a completely off-diagonal (in the unperturbed eigenbasis) perturbation, since this can always be achieved by absorbing the diagonal part of $V$ into $H_0$.
In particular this means that $\bl0V\kl0$ vanishes, so any Bloch sequence that ends (and/or starts) in zero and/or contains two zeroes in succession does not contribute to the eigenvector (energy) expansion.
As already noted by Salzman this greatly reduces the number of necessary diagrams~\cite{Salzman68}.

\section{Practical demonstration of diagrammatics\label{example}}
\subsection{Diagrammatic interpretation of recurrence relations\label{dRR}}
To facilitate a more thorough understanding of \cref{ceRR}, here we recount the recurrence relations diagrammatically.

Broadly speaking, the recurrence relations in \cref{RSPTRR,ceRR} both define how to compute higher order terms from lower order terms.
If we consider them in terms of diagrams, there is a key difference though.
\Cref{RSPTRR} defines how to construct (the sum of) all order $n$ diagrams by combining lower order diagrams.
On the other hand, \cref{ceRR} gives the coefficients of a single order $n$ diagram by deconstructing it into all possible compositions of lower order diagrams. Thus there is an implicit change of approach in going from the proof of \cref{RSPTRR} to \cref{ceRR}.
In the following, the point of view of \cref{ceRR} is illustrated more transparently.
We show at which points (marked with dots in \cref{dRRe,dRR0,dRR1,dRRg1}) diagrams should be cut in two and how.

\begin{figure}[h!]
\[e\left(\,
\begin{tikzpicture}[scale=0.5,baseline={([yshift=-\the\dimexpr\fontdimen22\textfont2\relax] current bounding box.center)},]
\draw (0,1)--(5,6)--(6,6)--(0,0) (2,4)--(4,4);
\draw[dashed] (0,0)..controls (0,1)..(1,1.5)..controls (2,2) and (1,3)..(2,4) (4,4)..controls (4,6)..(5,6);
\draw[fill=black] (3,4) circle[radius=0.1cm];
\end{tikzpicture}
\,\right)=c\left(\,
\begin{tikzpicture}[scale=0.5,baseline={([yshift=-\the\dimexpr\fontdimen22\textfont2\relax] current bounding box.center)},]
\draw (5,6)--(6,6)--(0,0) (2,4)--(4,4);
\draw[dashed] (0,0)..controls (0,1)..(1,1.5)..controls (2,2) and (1,3)..(2,4) (4,4)..controls (4,6)..(5,6);
\draw[white,very thick] (3,3.9)--(3,4.1);
\end{tikzpicture}
\,\right)-\sum_{\substack{\text{horizontal}\\ \text{intersections}}}c\left(\,
\begin{tikzpicture}[scale=0.5,baseline={([yshift=-\the\dimexpr\fontdimen22\textfont2\relax] current bounding box.center)},]
\draw (3,4)--(4,4)--(0,0)--(1,0);
\draw[dashed,xshift=1cm] (0,0)..controls (0,1)..(1,1.5)..controls (2,2) and (1,3)..(2,4);
\begin{pgfinterruptboundingbox}
\clip (0,-1) rectangle (4,5);
\draw[fill=black] (0,0) circle[radius=0.1cm] (4,4) circle[radius=0.1cm];
\end{pgfinterruptboundingbox}
\end{tikzpicture}
\,\right)e\left(\,
\begin{tikzpicture}[scale=0.5,baseline={([yshift=-\the\dimexpr\fontdimen22\textfont2\relax] current bounding box.center)},]
\draw[dashed,xshift=-4cm,yshift=-4cm] (4,4)..controls (4,6)..(5,6);
\draw (1,2)--(2,2)--(0,0);
\end{tikzpicture}
\,\right)\]
\caption{Diagrammatic illustration of the recurrence relation for $e$, \cref{eRR}. The sum is over all horizontal intersections with the upper diagonal, including (if applicable) the one at height $n$, in which case the argument of $e$ on the right-hand side is $\emptyset$. Here and in the following we mark the intersections we sum over and where we cut the diagrams with a dot.\label{dRRe}}
\end{figure}
For $e$, we take $c$ of the same diagram, then for every horizontal intersection with the upper diagonal we subtract a decomposition where we take $c$ of a diagram beginning with a $0$-step followed by everything before the intersection, multiplied by $e$ of the part of the original diagram following the $0$-step after the intersection, see \cref{dRRe}.

\begin{figure}[h]
\[c\left(\,
\begin{tikzpicture}[scale=0.5,baseline={([yshift=-\the\dimexpr\fontdimen22\textfont2\relax] current bounding box.center)},]
\draw[dashed] (1,0)..controls (1,1) and (3,1)..(3,3)..controls (3,5) and (4,4)..(4,5)..controls (4,6)..(5,6);
\draw (5,6)--(6,6)--(0,0)--(1,0);
\draw[fill=black] (3,3) circle[radius=0.1cm];
\end{tikzpicture}
\,\right)=\frac12\sum_{\text{intersections}}(-1)^{\begin{tikzpicture}[scale=0.5,baseline={([yshift=-\the\dimexpr\fontdimen22\textfont2\relax] current bounding box.center)},]
\draw (0,0)--(0.707,0.707);
\draw[dashed] (-0.1465,0.3535)--(0.8535,0.3535);
\draw[white] (0,-0.1)--(0,-0.3);
\end{tikzpicture}}
c\left(\,
\begin{tikzpicture}[scale=0.5,baseline={([yshift=-\the\dimexpr\fontdimen22\textfont2\relax] current bounding box.center)},]
\draw[dashed] (0,0)..controls (0,2) and (2,2)..(2,3);
\draw (2,3)--(3,3)--(0,0);
\begin{pgfinterruptboundingbox}
\clip (-1,0) rectangle (1,1);
\draw[fill=black] (0,0) circle[radius=0.1cm];
\end{pgfinterruptboundingbox}
\end{tikzpicture}
\,\right) c\left(\,
\begin{tikzpicture}[scale=0.5,baseline={([yshift=-\the\dimexpr\fontdimen22\textfont2\relax] current bounding box.center)},]
\draw[dashed,xshift=-3cm,yshift=-3cm] (3,3)..controls (3,5) and (4,4)..(4,5)..controls (4,6)..(5,6);
\draw[xshift=-3cm,yshift=-3cm] (5,6)--(6,6)--(3,3);
\begin{pgfinterruptboundingbox}
\clip (-1,0) rectangle (1,1);
\draw[fill=black] (0,0) circle[radius=0.1cm];
\end{pgfinterruptboundingbox}
\end{tikzpicture}
\,\right)\]
\caption{$k_1=0$ rule of $c$ recurrence, \cref{cRR}. At the intersection with the diagonal the diagram is cut in two. The first part is rotated by $\pi$, i.e. read backwards. The pictorial exponent is to be read as $1$ if the diagram intersects the diagonal horizontally, and $0$ if the intersection is vertical.\label{dRR0}}
\end{figure}
We split up the $c$ recurrence relations, \cref{cRR}, into three parts again.
For $k_1=0$, see \cref{dRR0}, we sum over all intersections with the main diagonal. There is at least one such intersection, since we start below the diagonal and end above it.
The part of the diagram before the intersection is read backwards, or equivalently is rotated by $\pi$.
The part after the intersection is left as is.
We multiply $c$ of both diagram parts and divide by $2$.
Horizontal intersections get a minus sign. (The example in \cref{dRR0} shows a vertical intersection.)

\begin{figure}[h]
\[c\left(\,
\begin{tikzpicture}[scale=0.5,baseline={([yshift=-\the\dimexpr\fontdimen22\textfont2\relax] current bounding box.center)},]
\draw[dashed] (1,1)..controls (1,2)..(2,2)..controls (4,2) and (3,5)..(4,5);
\draw (4,5)--(5,5)--(0,0)--(0,1)--(1,1);
\draw[fill=black] (1,1) circle[radius=0.1cm];
\end{tikzpicture}
\,\right)=c\left(\,
\begin{tikzpicture}[scale=0.5,baseline={([yshift=-\the\dimexpr\fontdimen22\textfont2\relax] current bounding box.center)},]
\draw[dashed,xshift=-1cm,yshift=-1cm] (1,1)..controls (1,2)..(2,2)..controls (4,2) and (3,5)..(4,5);
\draw[xshift=-1cm,yshift=-1cm] (4,5)--(5,5)--(1,1);
\begin{pgfinterruptboundingbox}
\clip (-1,0) rectangle (1,1);
\draw[fill=black] (0,0) circle[radius=0.1cm];
\end{pgfinterruptboundingbox}
\end{tikzpicture}
\,\right)\]
\caption{$k_1=1$ rule of $c$ recurrence, \cref{cRR}. Adding or removing a $k_1=1$ step in front of a diagram (or at any position) does not change its $c$-value.\label{dRR1}}
\end{figure}
The $k_1=1$ rule remains the simplest.
If a diagram starts with a $1$-step, remove it, see \cref{dRR1}.
This easily generalises to: remove all $k_i=1$-steps.
Though we should remember to stop at $n=1$; alternatively one could define $c(\emptyset)=1$, which would effectively make $\kl0$ the starting point instead of $\kl1$.

\begin{figure}[h]
\[c\left(\,
\begin{tikzpicture}[scale=0.5,baseline={([yshift=-\the\dimexpr\fontdimen22\textfont2\relax] current bounding box.center)},]
\draw[dashed] (0,2)..controls (0,3) and (1,2)..(1,3)..controls (1,4) and (2,3)..(2,4) (4,4)..controls (4,6)..(5,6);
\draw (0,2)--(0,0)--(6,6)--(5,6) (0,1)--(5,6) (2,4)--(4,4);
\draw[fill=black] (3,4) circle[radius=0.1cm];
\draw (-0.1,0.4)--(0.1,0.6) (3.4,3.9)--(3.6,4.1);
\end{tikzpicture}
\,\right)=\sum_{\substack{\text{horizontal}\\ \text{intersections}}} c\left(\,
\begin{tikzpicture}[scale=0.5,baseline={([yshift=-\the\dimexpr\fontdimen22\textfont2\relax] current bounding box.center)},]
\draw[dashed,yshift=-1cm] (0,2)..controls (0,3) and (1,2)..(1,3)..controls (1,4) and (2,3)..(2,4);
\draw (0,1)--(0,0)--(3,3)--(2,3);
\begin{pgfinterruptboundingbox}
\clip (0,0) rectangle (3,4);
\draw[fill=black] (3,3) circle[radius=0.1cm];
\end{pgfinterruptboundingbox}  
\end{tikzpicture}
\,\right) e\left(\,
\begin{tikzpicture}[scale=0.5,baseline={([yshift=-\the\dimexpr\fontdimen22\textfont2\relax] current bounding box.center)},]
\draw[dashed,xshift=-4cm,yshift=-4cm] (4,4)..controls (4,6)..(5,6);
\draw (1,2)--(2,2)--(0,0);
\end{tikzpicture}
\,\right)\]
\caption{$k_1>1$ rule of $c$ recurrence, \cref{cRR}. Similarly to \cref{dRRe}, a diagram is cut in two at every horizontal intersection with the upper diagonal. The first upwards unit towards the upper diagonal and the $0$-step after the intersection are discarded.\label{dRRg1}}
\end{figure}
For $k_1>1$, we have a sum over horizontal intersections with the upper diagonal, see \cref{dRRg1}.
The decomposition of the diagram has similarities with the one in the $e$-recurrence (\cref{dRRe}).
The second part of the diagram is treated the same but in the first part, instead of adding a $0$-step the first step is lowered by $1$.
Another difference is that here we are guaranteed to have at least one summand since the diagram starts above the upper diagonal.

Again a decomposition based on horizontal intersections with the upper diagonal results in one diagram where the upper diagonal becomes the main diagonal and (up to) one diagram containing the remainder.

As an example, consider
\begin{equation}
c(2,0,0,2)\overset{k_1>1}{=}c(1)e(0,2)\overset{e}{=}c(1)c(0,2)\overset{k_1=0}{=}\frac12c(1)^3=\frac12.
\end{equation}

\subsection{Diagrams up to fourth order}
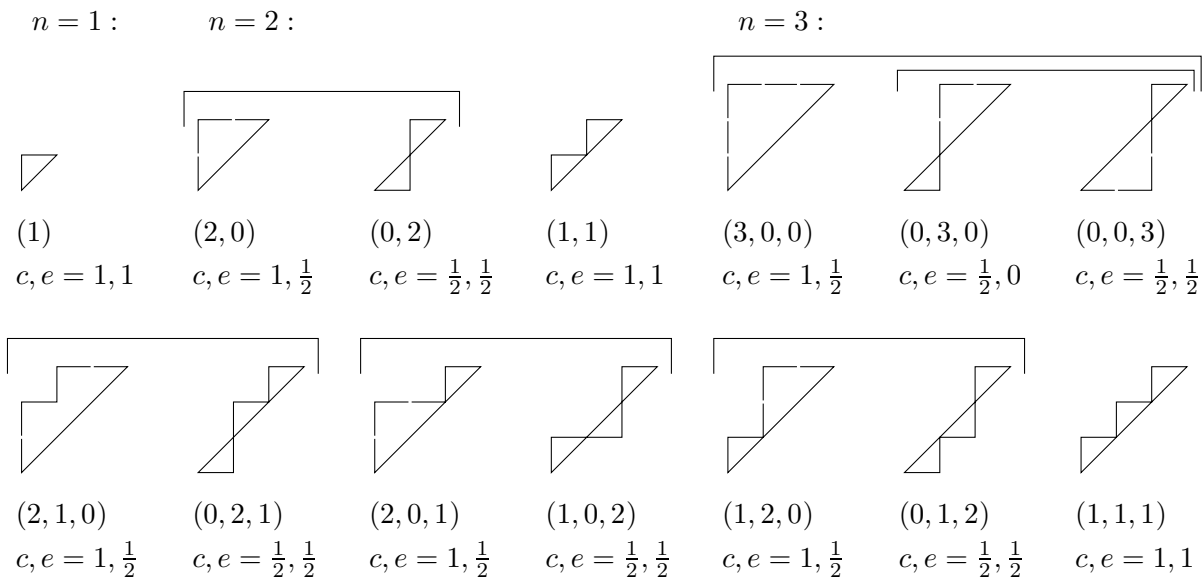
\begin{figure}
\begin{tikzpicture}[scale=0.4695]
\draw	(0,4.8) node[right]{$n=1:$}
		(5,4.8) node[right]{$n=2:$}
		(20,4.8) node[right]{$n=3:$};
\diagce{(0,0)}{{1}}11
\diagce{(5,0)}{{2,0}}1{\frac12}
\diagce{(10,0)}{{0,2}}{\frac12}{\frac12}
\eequiv522
\diagce{(15,0)}{{1,1}}11
\diagce{(20,0)}{{3,0,0}}{1}{\frac12}
\diagce{(25,0)}{{0,3,0}}{\frac12}0
\diagce{(30,0)}{{0,0,3}}{\frac12}{\frac12}
\eequiv{20}33 \cequiv{25}32

\begin{scope}[yshift=-8cm]
\diagce{(0,0)}{{2,1,0}}1{\frac12}
\diagce{(5,0)}{{0,2,1}}{\frac12}{\frac12}
\eequiv032
\diagce{(10,0)}{{2,0,1}}1{\frac12}
\diagce{(15,0)}{{1,0,2}}{\frac12}{\frac12}
\eequiv{10}32
\diagce{(20,0)}{{1,2,0}}1{\frac12}
\diagce{(25,0)}{{0,1,2}}{\frac12}{\frac12}
\eequiv{20}32
\diagce{(30,0)}{{1,1,1}}11
\end{scope}
\end{tikzpicture}
\caption{Bloch diagrams for orders $1$ through $3$, labelled with their corresponding Bloch sequences and $c$ and $e$ values. By \cref{permutations}, the horizontal square brackets group diagrams that contribute the same operator content to the expansion of the eigenenergy. Whenever there are double brackets, the inner ones indicate diagrams that contribute the same operator content to the eigenvector expansion. We can now compute effective coefficients $c_\text{eff}$, $e_\text{eff}$ by summing the bracketed $c$, $e$ and assign them to the left-most diagram (the canonical representative).\label{diagrams123}}
\end{figure}
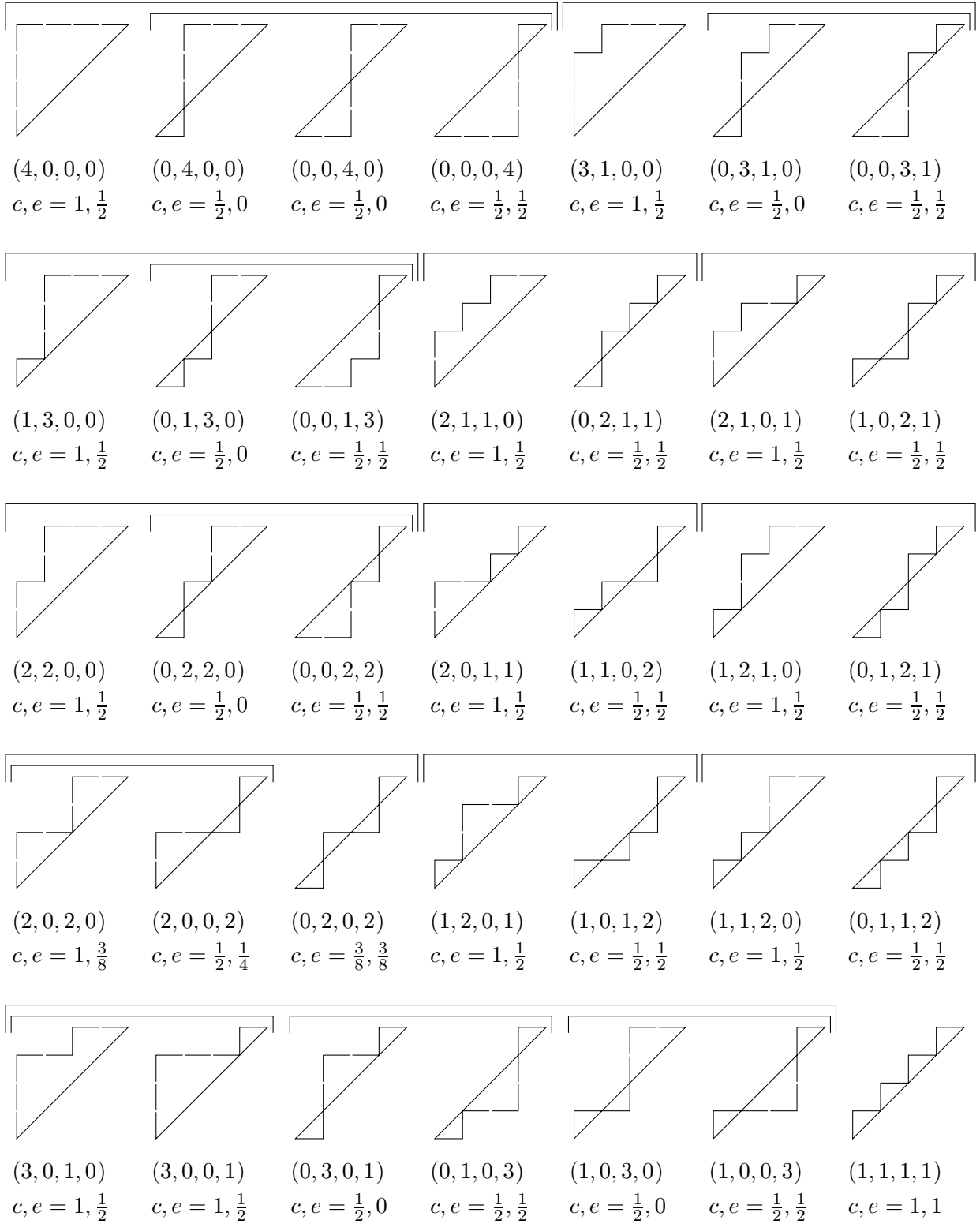
\begin{figure}
\begin{tikzpicture}[scale=0.4695]
\draw	(0,5.8) node[right]{$n=4:$};
\diagce{(0,0)}{{4,0,0,0}}1{\frac12}
\diagce{(5,0)}{{0,4,0,0}}{\frac12}0
\diagce{(10,0)}{{0,0,4,0}}{\frac12}0
\diagce{(15,0)}{{0,0,0,4}}{\frac12}{\frac12}
\eequiv044 \cequiv543
\diagce{(20,0)}{{3,1,0,0}}1{\frac12}
\diagce{(25,0)}{{0,3,1,0}}{\frac12}0
\diagce{(30,0)}{{0,0,3,1}}{\frac12}{\frac12}
\eequiv{20}43 \cequiv{25}42

\begin{scope}[yshift=-9cm]
\diagce{(0,0)}{{1,3,0,0}}1{\frac12}
\diagce{(5,0)}{{0,1,3,0}}{\frac12}0
\diagce{(10,0)}{{0,0,1,3}}{\frac12}{\frac12}
\eequiv043 \cequiv542
\diagce{(15,0)}{{2,1,1,0}}1{\frac12}
\diagce{(20,0)}{{0,2,1,1}}{\frac12}{\frac12}
\eequiv{15}42
\diagce{(25,0)}{{2,1,0,1}}1{\frac12}
\diagce{(30,0)}{{1,0,2,1}}{\frac12}{\frac12}
\eequiv{25}42

\begin{scope}[yshift=-9cm]
\diagce{(0,0)}{{2,2,0,0}}1{\frac12}
\diagce{(5,0)}{{0,2,2,0}}{\frac12}0
\diagce{(10,0)}{{0,0,2,2}}{\frac12}{\frac12}
\eequiv043 \cequiv542
\diagce{(15,0)}{{2,0,1,1}}1{\frac12}
\diagce{(20,0)}{{1,1,0,2}}{\frac12}{\frac12}
\eequiv{15}42
\diagce{(25,0)}{{1,2,1,0}}1{\frac12}
\diagce{(30,0)}{{0,1,2,1}}{\frac12}{\frac12}
\eequiv{25}42

\begin{scope}[yshift=-9cm]
\diagce{(0,0)}{{2,0,2,0}}1{\frac38}
\diagce{(5,0)}{{2,0,0,2}}{\frac12}{\frac14}
\diagce{(10,0)}{{0,2,0,2}}{\frac38}{\frac38}
\eequiv043 \cequiv042
\diagce{(15,0)}{{1,2,0,1}}1{\frac12}
\diagce{(20,0)}{{1,0,1,2}}{\frac12}{\frac12}
\eequiv{15}42
\diagce{(25,0)}{{1,1,2,0}}1{\frac12}
\diagce{(30,0)}{{0,1,1,2}}{\frac12}{\frac12}
\eequiv{25}42

\begin{scope}[yshift=-9cm]
\diagce{(0,0)}{{3,0,1,0}}1{\frac12}
\diagce{(5,0)}{{3,0,0,1}}1{\frac12}
\cequiv042
\diagce{(10,0)}{{0,3,0,1}}{\frac12}0
\diagce{(15,0)}{{0,1,0,3}}{\frac12}{\frac12}
\cequiv{10}42
\diagce{(20,0)}{{1,0,3,0}}{\frac12}0
\diagce{(25,0)}{{1,0,0,3}}{\frac12}{\frac12}
\eequiv046 \cequiv{20}42
\diagce{(30,0)}{{1,1,1,1}}11
\end{scope}
\end{scope}
\end{scope}
\end{scope}
\end{tikzpicture}
\caption{Bloch diagrams for order $4$, arranged analogously to \cref{diagrams123}.\label{diagrams4}}
\end{figure}
\Cref{diagrams123,diagrams4} show all Bloch diagrams up to fourth order with their $c$ and $e$ values.
Diagrams producing equivalent operator content are grouped together.
One can easily verify that the results for the energy are consistent with Bloch's, \cref{energybloch}, by summing up all the grouped $e$ coefficients and getting the number of convex diagrams in the group.

Fourth order perturbation theory is not exactly an outlandish endeavour, yet has sufficient complexity that even though the associated Talk page has since 2010 noted that there are mistakes in the expressions listed on Wikipedia, to date no one has corrected them~\cite{wiki}.
There are $35$ Bloch sequences for $n=4$, of which $14$ are convex, $13$ need to appear in the energy series ($4$ if $V$ completely off-diagonal), and $26$ need to appear in the normalised eigenvector series ($12$ if $V$ is completely off-diagonal).

\section{Conclusion\label{conclusio}}
We have shown how to explicitly solve the conventionally normalised Rayleigh-Schr\"odinger perturbation series to arbitrary order.
The structure of earlier results for unnormalised vectors is readily adapted to this problem.
The normalisation necessarily increases the number of terms in the expansion.
We surveyed how the number of terms varies between different methods, and how to identify equivalent diagrams.
An efficient summation of these equivalent diagrams remains an open problem, and there is likely no simple solution.

No matter how efficiently terms are summarised, their number grows exponentially with the order of perturbation.

Counting and analysing Bloch diagrams and associated quantities offers a rich trove of combinatorics problems, many of which may have already been studied in the context of paths, random walks, and bridges.

\begin{acknowledgments}
We thank the OpenSuperQ project (820363) of the EU Flagship on Quantum Technology, H2020-FETFLAG-2018-03, for support.
\end{acknowledgments}

\bibliography{21feb_arxiv}

\end{document}